\documentclass[11pt]{article}
\usepackage{fullpage}
\usepackage[utf8]{inputenc}
\usepackage[english]{babel}
\usepackage[T1]{fontenc}
\usepackage{amsmath,amssymb,amsfonts,amsthm}
\usepackage{url}
\usepackage{hyperref}
\usepackage{cleveref}
\usepackage{graphicx}
\usepackage{makecell}
\usepackage{algorithm}
\usepackage{algpseudocode}
\usepackage{tikz}
\usepackage{physics}
\usepackage{mathcommands}

\usepackage[mode=buildnew,subpreambles=true]{standalone}

\usepackage[backend=bibtex,doi=false,isbn=false,url=false,maxbibnames=5,sorting=none]{biblatex}
\AtEveryBibitem{\clearfield{month}}
\AtEveryBibitem{\clearfield{day}}
\newtheorem{theorem}{Theorem}
\newtheorem*{theorem*}{Theorem}
\newtheorem{lemma}[theorem]{Lemma}
\newtheorem{proposition}[theorem]{Proposition}

\newtheorem{corollary}[theorem]{Corollary}

\newtheorem{definition}[theorem]{Definition}

\usepackage{authblk}

\newcommand{\card}[1]{\lvert #1 \rvert}

\usepackage{xcolor}

\title{On combinatorial structures in linear codes}

\author[1]{Nou\'edyn Baspin}

\affil[1]{Centre for Engineered Quantum Systems, School of Physics, University of Sydney, Sydney, NSW 2006, Australia}

\begin{document}
\maketitle

\begin{abstract}
   In this work we show that given a connectivity graph $G$ of a $[[n,k,d]]$ quantum code, there exists $\{K_i\}_i, K_i \subset G$, such that $\sum_i |K_i|\in \Omega(k), \  |K_i| \in \Omega(d)$, and the $K_i$'s are $\Omegalog( \sqrt{{k}/{n}})$-expander. If the codes are classical we show instead that the $K_i$'s are $\Omegalog\left({{k}/{n}}\right)$-expander. We also show converses to these bounds. In particular, we show that the BPT bound for classical codes is tight in all Euclidean dimensions. Finally, we prove structural theorems for graphs with no "dense" subgraphs which might be of independent interest.
\end{abstract}

\setcounter{tocdepth}{2}
\tableofcontents

\section{Introduction}
	
	A linear\footnote{Any classical code mentioned henceforth will be assumed to be linear.} classical code on $n$ bits encoding $k$ logical bits is defined as a subspace $\cC \cong \F_2^k$ of $\F_2^n$. Given any such subspace, it is possible to find by gaussian elimination a matrix $\pcm : \F_2^n \rightarrow \F_2^{n-k}$ such that $\cC = \ker \pcm$. We say that $\code$ has distance $d$ if all $c \in \code$ satisfy $|c| \geq d$, where $|c|$ is the Hamming weight of $c$.
	Those rows span $\cC^\perp \subset (\F_2^n)^*$ the subspace of the dual space such that  $c \in \code$ iff for for all $v \in \code^{\perp}$, we have $\langle v, c\rangle = 0$. These rows -- also called the "checks" of a code -- have the natural interpretation of being a measurement: $\langle v, c\rangle = 0$ iff the values of the bits indexed by $v$ sum to $0$. The set of the outcomes of those measurements for a vector $c$ are called the \emph{syndrome} of $c$, and corresponds to the vector $\pcm c$.

    Error correcting codes an also be defined for quantum systems. For a finite-dimensional system with $n$ identical particles, and Hilbert space $\hbt^{\otimes n}$, a quantum code is a subspace $\code \subset \hbt^{\otimes n}$. One says that $\cC$ protects $k$ logical qubits if $\dim(\cC) = 2^k$, and has distance $d$ if for any $\ket{\psi} \in \code$, one can recover $\ket{\psi}$ from $\tr_{A} \dyad{\psi}{\psi}$ whenever $A$ consists of strictly less than $d$ particles. In this work we will focus on stabilizer which correspond to the quantum analog of classical linear codes. Given an abelian subgroup $\cS$ of the $n$-qubit Pauli group, a stabilizer code $\cC$ corresponds to the common $+1$ eigenspace of these operators. In that case, a set of checks for this code is a set of Pauli operators that generate
    $\cS$, and can also be understood as a set of measurement one can perform to verify whether the system is in a codestate $\ket{\psi} \in \cC$.

    It is then natural to consider the following question: is it possible to produce codes with good performances while only having access to a restricted class of measurements. In \cite{bravyi2010tradeoffs} it was shown that given a classical code $\cC$ whose checks be locally embedded into $\E^D$, then this code has to obey $k d^{1/D}\in O(n)$. One can verify that saturating this bound is equivalent to finding an $\param{n}{n^{(D-1)/D}}{ n}$ family of codes. This bound is saturated by the repetition code in 1D, however, to the best knowledge of the author, the problem of saturating this bound in dimensions $D \geq 2$ is open.

    Similar bounds also exist for stabilizer codes: this kind of code when their checks are local in $\E^D$ has to obey $k d^{2/(D-1)}\in O(n)$, $d \in O(n^{(D-1)/D})$. In which case, saturating these bounds correspond to finding a family of $[\param{n}{n^{(D-2)/D}}{ n^{(D-1)/D}}]$. Showing that for all $D \geq 1$ those bounds are tight -- or alternatively finding tighter bounds -- is an open question in quantum error correction. The existence of those bounds, and the difficulty to saturate or tighten them, seems to imply that the interplay between locality and error correcting codes is non-trivial: for example, what is the unique property of $\E^3$ that might allow us to build better codes than in $\E^2$?
    
    In \cite{baspin2021connectivity, baspin2023improved} , the authors use some graph theoretical tools -- primarily the \emph{separation profile} --  to extend the BPT bounds to codes not living in $\E^D$. These bounds, although they do not allow to completely recover the original BPT bounds, have the benefit of being applicable to {any} code, in a geometry-oblivious fashion.
    Although this approach works very well as long as the separation profile is polynomial, it quickly becomes unpractical when it is not. 
     
    In this paper we will argue that instead the notion of expansion is a more natural metric, and applies more universally. As a starting point, we will consider Proposition 17 of the Supplemental Material of \cite{baspin2021quantifying}, which guarantees the existence of some expanders in any connectivity graph of any code with high $k,d$. For a given set of checks for a classical or quantum code $\cC$, the connectivity graph has the (qu)bits of $\cC$ as vertices, and two vertices are linked by an edge if the associated (qubits) are in the support of the same check. The proposition reads:
	
	\begin{proposition}
		
		There exists constants $\beta, n_0$, such that for any quantum code $[[n,k,d]]$ satisfying $n \geq n_0$, let $G$ be a connectivity graph for the code, then we have
		
		\begin{enumerate}
			\item There exists $ K \subset G$, such that $|K| \geq d/2$, and $K$ is $\Omega(d/n)$-expander.
			\item There exists $K \subset G$, such that $|K| \geq \frac{d}{2}\sqrt{\frac{k}{\beta n\log(n)^2}}^{1/\log_n(d)}$, and $K$ is $\frac{1}{3}\sqrt{\frac{k}{\beta n\log(n)^2}}^{1/\log_n(d)}$-expander.
			\item If $k(n)d(n)^{2/D} \geq \beta n\log(n)^2/(1-\alpha) $, for $\alpha \in (0,1)$, then there exists $\{K_i\}_i, K_i \subset G$, such that $\sum_i |K_i| > \alpha k \sqrt{\frac{(1-\alpha)k}{\beta n\log(n)^2}}^{1/\log_n(d)}$, $|K_i| \geq  \frac{d}{2}\sqrt{\frac{(1-\alpha)k}{\beta n\log(n)^2}}^{1/\log_n(d)}$, and $K_i$ is $\frac{1}{3}\sqrt{\frac{(1-\alpha)k}{\beta n\log(n)^2}}^{1/\log_n(d)}$-expander. 
		\end{enumerate}
		
	\end{proposition}
	
	That this proposition is a bit inelegant to state reflects the clunkiness that comes with using the separation profile. While the first item is relatively "clean", there is clearly room for improvement regarding the other two. In particular we might expect a tighter, more general statement to hold:

	\begin{proposition}
		
		For any $[[n,k,d]]$ quantum code, let $G$ be a connectivity graph for the code. There exists $\{K_i\}_i, K_i \subset G$, such that 
		
		\[
		\sum_i |K_i|\in \Omega(k), \  |K_i| \in \Omega(d)
		\]
		
		and $K_i$ is $\Omega\left( \sqrt{\frac{k}{n}}\right)$-expander. 
		
	\end{proposition}

    Another hint that expansion might be a useful framework to derive bounds on codes can be found by making a few observations regarding  \cite{baspin2021quantifying}. Assume that $G$ is $t$-dense, there exists $H \subset G$ such that $H$ is $t/|H|$ expander. So there is a drawback to $H$ being small: the expansion goes up. If $H$ has expansion $t/|H|$, then it induces $|H|$ edges of length $t/\sqrt{|H|}$, and there is a tradeoff between $H$ being large, and inducing many edges versus $H$ being small and inducing long edges. Of course we can take the worse case scenario, i.e. we always have $n \geq |H| \geq t$, with $n$ the size of $G$. This yields least $t$ edges of size $t/\sqrt{n}$, but is clearly suboptimal. We hope that focusing on expansion might circumvent these drawbacks.

	In this note, we make a step by showing that the following proposition holds:
	
	\begin{proposition}
		
		For any $[[n,k,d]]$ quantum code, let $G$ be a connectivity graph for the code. There exists $\{K_i\}_i, K_i \subset G$, such that 
		
		\[
		\sum_i |K_i|\in \Omega(k), \  |K_i| \in \Omega(d)
		\]
		
		and $K_i$ is $\Omegalog\left( \sqrt{\frac{k}{n}}\right)$-expander. 
		
	\end{proposition}

	nd a classical equivalent is also shown in Corollary \ref{cor:large-expanders}. This result is certainly still far from optimal since we still do not make use of the expansion lemma, we thus expect it to be less tight than \cite{baspin2023improved}. Unlike the quantum case, the bounds we obtain on classical codes allow us to retrieve the BPT bound up to polylog factors. 
    
    Finally, we also present converses to those results: we show how the BPT bound for classical codes is tight in Section \ref{sec:converse-bpt}, and we show a partial converse to the more general bounds from expansion in Section \ref{sec:converse-exp}. There is history of work attempting to saturate this kind of bound, which we summarize below.
	
	\begin{enumerate}
		\item In \cite{bravyi2010tradeoffs}, an attempt is made to saturate the classical BPT bound by numerically generating codes from cellular automata. For example consider a one dimensional line of bits of length $l$, and apply a cellular automaton during $\Delta$ steps. Then the codewords are defined on $l*(\Delta + 1)$ bits, and correspond to the $(\Delta + 1)$ states of the line during the computation plus the initial state. This construction however does not saturate the bound, and does provide an analytical lower bound on the distance of the code. 
		
		\item In \cite{yoshida2013information} a construction similar to that of \cite{bravyi2010tradeoffs} is considered, but focuses on fractal-generating automata. For codes on bits, the author describes a $2D$ local family with parameters $\param{n}{n^{1/2}}{n^{\sim0.8}}$, and a $3D$ local family with parameters $\param{n}{n^{2/3}}{n^{2/3}}$. However, one can formulate a version of the classical BPT bound for higher dimensional spins, and this construction asymptotically saturates it when the spin dimension diverges.
		
		\item In \cite{bacon2015sparse} the authors show that by considering a syndrome extracting circuit, rather than a cellular automaton, a construction similar to that of \cite{bravyi2010tradeoffs} yields subsystem codes with provable lower bounds. In particular, they show that starting with an LDPC stabilizer code with parameters $[n_0,k_0,d_0]$, then one obtains a $D$-local subsystem code on $n \in \Omega(n_0)$ qubits with parameters $\param{n}{k_0/n_0\cdot n^{(D-1)/D}}{d_0/n_0\cdot n^{(D-1)/D}}$. This construction is optimal for subsystem codes since, when used with good LDPC stabilizer codes \cite{panteleev2021good,leverrier2022quantum}, they yield codes with parameters $\param{n}{n^{1-1/D}}{n^{1-1/D}}$, which is best possible for all ranges of $k,d$.
	\end{enumerate}

\subsection{Open Questions}

The lack of tightness of our bounds can probably in large part be attributed to very ad-hoc graph-theoretical constructions. It makes sense to raise the question of their optimality:

\begin{enumerate}
	\item Expansion concentration lemma \ref{prop:expansion-concentration}: can we improve the expansion of $G[U]$ from $\Omegalog\left(\epsilon\right)$ to $\Omega(\epsilon)$?
	\item Partitioning lemma \ref{lem:partitioning}: Can we show that the regions $A'_j$ also have at most $O(\epsilon m)$ boundary edges? That would allow us to make use of the expansion lemma and recover the results of \cite{baspin2023improved}.
\end{enumerate}

\section{Graph partitioning results}

The goal of that section is to prove the following intuitive section: if a graph $G$ does not contain any "dense" subgraph $G' \subset G$, then $G$ can easily be partitioned. The formal result can be found in Theorem \ref{thm:no-expander-partitioning}. We start by introducing the graph theoretical tools we will need. 

\begin{definition}
		Let $G = (V,E)$ be a graph. A separator is a subset $S \subset E$ of edges such that its removal leaves two disconnected subgraphs $G_1 = (V_1, E_1), G_2 = (V_2, E_2)$ such that $|V_1|, |V_2| \leq \frac{2}{3}n$. We can now define $\sep(G)$: 
		
		\[
		\sep(G) = \min_{S \text{ is a separator of } G'} |S|
		\]
	\end{definition}

	\begin{definition}
		Let $G = (V,E)$ a finite graph. For any $U \subset V$, we define 
		
		\[
		\phi_G(U) = \frac{\card{\bdry U}}{\card{U}}
		\]
		
		where $\bdry U \subset E$ is the set of edges with exactly one endpoint in $U$. The small scale expansion constant $h_\alpha(G)$ is defined as
		
		\[
		\min_{U \subset V, \card{U} \leq \alpha \card{V}} \phi_G(U)
		\]
		
		And we say that $G$ is $\epsilon$-expander if 
		
		\[
		h_{1/2}(G) \geq \epsilon
		\]

	\end{definition}

	At times it will be useful to consider a subset of $\bdry U$. For some $W \subset V$, we write $\bdry^W U$ the set of edges that have one endpoint in $U$, and the other in $W$. We write $G[U] = (U, E_U)$ is the subgraph induced by $U$ where $E_U \subset E$ denotes the edges in $E$ that have both endpoints in $U$.

    We will also make use of the following Lemma and Corollary which relate separation and expansion.

\begin{lemma}[Inspired by \cite{bottcher2010bandwidth}, Lemma 12]
		Let $G$ be a graph on $n$ vertices. If there is no $U \subset G$ with $\card{U} \geq 2n/3$ such that $G[U]$ is $\epsilon$-expander then $\sep(G) \leq \frac{2}{3}\epsilon n$.
	\end{lemma}
	\begin{proof}
		We recursively define three sequences $(A_i)_i, (S_i)_i, (B_i)_i$.
		Let $A_0 = S_0 = \emptyset$, $B_0 = V$.
		As long as $|B_i| \geq \frac{2}{3}n$ then $G[B_i]$ is not an $ \epsilon$-expander, and we can find a set $A_{i+1} \subset B_i, \card{A_{i+1}} \leq \card{B_i}/2$ such that $\phi_{G[B_i]}(A_{i+1}) \leq \epsilon$. We let $S_{i+1} \equiv \bdry^{B_i} A_{i+1}$. 
		Write $l$ the first index for which $\card{B_l} \leq \frac{2}{3}n$. We then have, by definition $\card{B_{l}} = \card{B_{l-1}} - \card{A_{l}} \geq \card{B_{l-1}} - \frac{1}{2}\card{B_{l-1}}  = \frac{1}{2}\card{B_{l-1}} \geq \frac{1}{3}n$.
		
		We obtain $A' \equiv \bigcup_{i=1}^{i=l} A_i$, $B' \equiv B_l$, $S' \equiv \bigcup_{i=1}^{i=l} S_i$, where $A' \sqcup B'$ is a partition of $V$, and removing the edges $S'$ disconnects the vertices of $A'$ from the vertices of $B'$. Since $\frac{1}{3}n \leq B' \leq \frac{2}{3}n$, then $\frac{1}{3}n \leq A' \leq \frac{2}{3}n$, we conclude that $S'$ is a separator for $G$. Further we can verify that 
		\[
		\card{S'} = \sum_i \card{\bdry^{B_i} A_{i+1}} = \sum \phi_{G[B_i]}(A_{i+1}) \card{A_{i+1}} \leq \epsilon \sum_i \card{A_{i+1}} = \epsilon \card{A'} \leq \frac{2}{3} \epsilon n \ .
		\]
		
		This gives $\sep(G) \leq \frac{2}{3} \epsilon n$.
	\end{proof}
	
	\begin{corollary}
		\label{cor:high-sep-high-exp}
		Let $G$ be a graph on $n$ vertices. If $\sep(G) \geq \epsilon n$, then there exists $U \subset G$ with $\card{U} \geq 2n/3$ such that $G[U]$ is $\frac{3}{2}\epsilon$-expander.
	\end{corollary}
	\begin{proof}
		Contrapositive of the previous lemma.
	\end{proof}

 \subsection{Expansion concentration}

Let $G= (V,E)$ be a graph on $n$ vertices. A subset $U \subset V$ is said to be $\alpha$-small if $\card{U} \leq \alpha \card{V}$, $\alpha$-big if $\card{U} \geq \alpha \card{V}$, and $\epsilon$-Folner if $\frac{\card{\bdry U}}{\card{U}}\leq \epsilon$.

In this section we show that if a graph $G$ contains no $\alpha$-small $\epsilon$-Folner set, then there is a subgraph $G' = (V',E') \subset G$, where $V'$ is $\frac{2}{3}\alpha$-big, such that $G'$ contains no $\frac{1}{2}$-small $\Omega(\frac{\epsilon}{\log(n)})$-Folner set. In some sense we start a graph $G$ with some well defined but "dilute" expansion that is guaranteed only for small subsets, and we concentrate it into a subgraph $G'$ for which the expansion is almost the same, and valid for all $\frac{1}{2}$ subsets, which constitutes the usual notion of expansion.

	\begin{proposition}[Expansion concentration]
        \label{prop:exp-concentration}
		Let $G = (V,E)$ be a graph on $n$ vertices satisfying 
		
		\[
		h_{\alpha}(G) \geq  \epsilon
		\]
		
		then there exists $U \subset V, \card{U} \geq \frac{2}{3} \alpha \card{V}$, such that 
		
		\[
		h_{1/2}(G[U]) \geq \frac{\epsilon}{4 \log_{3/2}(1/\alpha)} \in \Omega\left(\frac{\epsilon}{\log(n)}\right)
		\]
	\end{proposition}
	\begin{proof}
		If $h_{1/2}(G) \geq \epsilon/2$ the lemma follows, so we will assume $h_{1/2}(G) \leq \epsilon/2$.  This assumption guarantees the existence of $U \subset V, \alpha \card{V} \leq \card{U} \leq \card{V}/2$ such that $\phi_G(U) \leq \epsilon/2$. 	We will recursively define a sequence $(Y_i)_i, Y_i \subset V$, and we take $Y_1 \equiv U$.
		
		By definition, for any partition $Y_i = A \sqcup B$, we have
		
		\begin{align*}
			\phi_G(Y_i) = \frac{\card{\bdry Y_i}}{\card{Y_i}}
			= \frac{\card{\bdry^{\comp{Y_i}}A} + \card{\bdry^{\comp{Y_i}}B}}{|A| + |B|}
		\end{align*}
	
		where $\comp{Y_i} \equiv V \setminus Y_i$.
		
	
		WLOG, we take $\frac{\card{\bdry^{\comp{Y_i}}A}}{\card{A}} \leq \frac{\card{\bdry^{\comp{Y_i}}B}}{\card{B}}$. Using a standard algebraic inequality, this gives
		
		\[
		\frac{\card{\bdry^{\comp{Y_i}}A}}{\card{A}}  \leq \phi_G(Y_i) \leq \frac{\card{\bdry^{\comp{Y_i}}B}}{\card{B}}
		\]
		
		Using the fact that $\bdry A  = \bdry^{\comp{Y_i}}A \cup \bdry^B A $, we can obtain:
		
		\[
		\frac{\card{\bdry A } - \card{\bdry^B A }}{\card{A }} = \frac{\card{\bdry^{\comp{Y_i}}A }}{\card{A }} \leq \phi_G(Y_i)
		\]
		
		It remains to specify our choice of $A,B$. Let $G[Y_i]$ have an edge separator $S_i$, then $Y_i = A \sqcup B$ with $\frac{1}{3}\card{Y_i} \leq |A| \leq \frac{2}{3}\card{Y_i}$. We choose $Y_{i+1} = A$, applying the previous inequality, this yields
		
		\[
		\phi_G(A ) \leq \phi_G(Y_i) + \frac{\card{S_i}}{\card{Y_{i+1} }} \ .
		\]
		
		We write $\delta_i \equiv \frac{\card{S_i}}{\card{Y_{i+1} }}$. Substituting $Y_{i+1} $ for $Y_i$, we repeat the argument recursively until the first $l$ such that $\card{Y_{l+1}} \leq \alpha \card{V}$. At that point, we have
		
		\[
		\epsilon \leq \phi_G(Y_{l+1}) \leq \phi_G(Y_l) + \delta_l = \phi_G(Y_1) + \sum_{i=1}^{i=l} \delta_i
		\]
		
		or
		
		\[
		\sum_i \delta_i \geq \epsilon/2
		\]
		
		Due to the upper bound on the size of $\card{Y_{i+1}} \leq \frac{2}{3}\card{Y_i}$, $l$ is at most $\lceil \log_{3/2}(n/\alpha n) \rceil$. Write $i' = \arg\max_i \delta_i$, then
		
		\[
		\lceil \log_{3/2}(\alpha^{-1}) \rceil \delta_{i'} \geq \epsilon/2
		\]
		
		or
		
		\[
		\delta_{i'} \geq \frac{\epsilon}{2 \lceil \log_{3/2}(\alpha^{-1}) \rceil}
		\]
		
		By the definition of a separator, we know that $\card{Y_{i'+1}} \geq \frac{1}{3}\card{Y_{i'}}$, which yields
		
		\[
		\frac{\card{S_{i'}}}{\frac{1}{3}\card{Y_{i'}}} \geq \frac{\card{S_{i'}}}{\card{Y_{i'+1}}} =  \delta_{i'} \geq \frac{\epsilon}{2 \lceil \log_{3/2}(\alpha^{-1}) \rceil}
		\]
		
		or 
		
		\[
		\card{S_{i'}} \geq \frac{\epsilon}{6 \lceil \log_{3/2}(\alpha^{-1}) \rceil} \card{Y_{i'}}
		\]
		
		By applying Corollary \ref{cor:high-sep-high-exp}, we can guarantee that there exists $U \subset Y_{i'}$ such that $\card{U} \geq \frac{2}{3}\card{Y_{i'}} \geq \frac{2}{3} \alpha \card{V}$, and $G[U]$ is $\frac{\epsilon}{4 \lceil \log_{3/2}(\alpha^{-1}) \rceil}$-expander.
	\end{proof}

\subsection{Partitioning non-expanding graphs at low cost}

    In this section, we prove our main graph-theoretical result: a graph with no "dense" subgraphs can be partitioned at low cost. For the sake of convenience, we will use a slightly different definition for the small scale expansion constant.

\begin{definition}
		Let $G = (V,E)$ be a finite graph. The small scale expansion constant $\hat{h}_m(G)$ is defined as
		
		\[
		\hat{h}_m(G) = \min_{U \subset V, \card{U} \leq m} \phi_G(U)
		\]
		
	\end{definition}
 
    Proposition \ref{prop:exp-concentration} can thenbe rephrased in the following terms.

	\begin{proposition}[Expansion concentration]
		\label{prop:expansion-concentration}
		Let $G = (V,E)$ be a graph on $n$ vertices satisfying 
		
		\[
		\hat{h}_{m}(G) \geq  \epsilon
		\]
		
		then there exists $U \subset V, \card{U} \geq \frac{2}{3} m$, and a universal constant $c$ such that $G[U]$ is $\frac{\epsilon}{c \cdot \log(n)}$-expander.
	\end{proposition}

	\begin{corollary}
		\label{cor:expansion-concentration}
		Let $G = (V,E)$ be a graph on $n$ vertices for which there is no $U \subset V$, $\card{U} \geq m$ such that $G[U]$ is expander, then 
		
		\[
		\forall G' \subset G, \ \hat{h}_{\frac{3}{2}m}(G') \leq  c \cdot \epsilon \cdot \log(n)
		\]
	\end{corollary}

	\begin{lemma}
		\label{lem:partitioning}
		Let $G$ such that for any $G'$ with $\card{G'} \geq m$ we have $\hat{h}_{m}(G') \leq \epsilon$, then can partition into regions $\{A'_j\}_{j=1}^{j=t}$, with $A'_j < 2m$, $t \leq n/m + 1$, and there are at most $\epsilon n $ edges with endpoints in different $A'_j$'s.
	\end{lemma}

	\begin{proof}
		We recursively define three sequences $(A_i)_i, (S_i)_i, (B_i)_i$.
		Let $A_0 = S_0 = \emptyset$, $B_0 = V$.
		As long as $|B_i| \geq m$ then $\hat{h}_{m}(G[B_i]) \leq \epsilon$, and we can find a set $A_{i+1} \subset B_i, \card{A_{i+1}} \leq m$ such that $\phi_{G[B_i]}(A_{i+1}) \leq \epsilon$. We let $S_{i+1} \equiv \bdry^{B_i} A_{i+1}$. 
		Write $l$ the first index for which $\card{B_l} \leq m$, and we pick $A_{l+1} \equiv B_l$.
		
		Now note that $V = \bigsqcup_{i=1}^{i=l+1} A_i $, and $\card{A_i} \leq m$. We can reorganize those $A_i$'s to form bigger (but no too big) sets $A'_j$. Indeed define 
		
		\[
		A'_1 = \bigcup_{i=1}^{i=l'_1} A_i \ ,
		\]
		
		where $l'_1$ is the smallest integer such that 
		
		\[
		\card{\bigcup_{i=1}^{i=l'_1+1} A_i} \geq 2m \ ,
		\]
		
		then one can verify that $m \leq A'_1 < 2m$. We can repeat this process a number of times $t$ until we exhaust all $A_i$'s. We thus obtain $\{A'_j\}_{j=1}^{j=t}$, where $V = \bigsqcup_{j=1}^{j=t} A'_j$. For $1 \leq j \leq t-1$ we have $m \leq \card{A'_j }< 2m$, and $\card{A'_t } \leq m$: we cannot lower bound the size of the last set. We now obtain an upper bound on $t$ by noting that
		
		\[
		n  \geq (t-1) \cdot \min_{j: 1 \leq j \leq t-1} \card{A'_j} \geq (t-1) \cdot m \ ,
		\]
		
		this gives 
		
		\[
		t \leq \frac{n}{m} + 1
		\]
		
		By definition all the edges between $A'_{j_1},A'_{j_2}$ are contained within $S' = \bigcup_{i=1}^{l} S_i$. We can bound its size
		 
		\[
		\card{S'} = \sum_{i=0}^{i=l-1} \card{\bdry^{B_i} A_{i+1}} = \sum_{i=0}^{i=l-1} \phi_{G[B_i]}(A_{i+1}) \card{A_{i+1}} \leq \epsilon \sum_{i=0}^{i=l-1} \card{A_{i+1}} = \epsilon \card{A'} \leq \epsilon n \ .
		\]
		
	\end{proof}

	Combining Corollary \ref{cor:expansion-concentration} and Proposition \ref{lem:partitioning}, we obtain a general partitioning theorem:

	\begin{theorem}
		\label{thm:no-expander-partitioning}
		Let $G = (V,E)$ be a graph on $n$ vertices for which there is no $U \subset V$, $\card{U} \geq m$ such that $G[U]$ is $\epsilon$-expander, then $G$ can be partitioned into $\{A_i\}_i$ such that 
		
		\begin{enumerate}
			\item $\card{A_i} < 3m$
			\item there at most $c \cdot \log(n) \cdot \epsilon  n $ edges with endpoints in different $A_i$'s
		\end{enumerate}
	\end{theorem}

\section{Expansion-based bounds on codes}

We are now in position to prove our main result. Using the framework of \cite{baspin2021connectivity,bravyi2010tradeoffs}, we can verify that the following lemma holds:

\begin{lemma}
    Let $G = (V,E)$ be a connectivity graph on $n$ vertices 

    \begin{enumerate}
        \item 
            for a classical code such that its vertices can be partitioned into $A \sqcup_i B_i$, such that there is no two $B_{i'}, B_{i''}$ sharing any edges, then 
    		
    		\begin{enumerate}
    			\item $d \leq  \max_i |B_i|$, or
    			\item $k \leq |A|$
    		\end{enumerate}
        \item 
            for a quantum code such that its vertices can be partitioned into $A \sqcup_i B_i \sqcup_j C_j$, such that there is no two $B_{i'}, B_{i''}$ sharing any edges, nor two $B_{j'}, B_{j''}$ sharing any edges, then 
    		
    		\begin{enumerate}
    			\item $d \leq  \max_i |B_i|, \max_j |C_j|$, or
    			\item $k \leq |A|$
    		\end{enumerate}
    \end{enumerate}
\end{lemma}

    This lemma, used with Theorem \ref{thm:no-expander-partitioning} gives the following result.
	\begin{theorem}
		\label{thm:no-expander-code-bounds}
		Let $G = (V,E)$ be a connectivity graph for a code $\cC$ on $n$ vertices for which there is no $U \subset V$, $\card{U} \geq m$ such that $G[U]$ is $\epsilon$-expander, then 

        \begin{enumerate}
            \item 
                if $\cC$ is a classical code, it obeys either
		
        		\begin{enumerate}
        			\item $d \leq 3m$, or
        			\item $k \leq 2c \cdot \log(n) \cdot \epsilon n$
        		\end{enumerate}
            \item 
                if $\cC$ is a quantum code, it obeys either

                \begin{enumerate}
        			\item $d \leq 3m$, or
        			\item $k \leq 2(c \cdot \log(n) \cdot \epsilon)^2 n$
        		\end{enumerate}

        \end{enumerate}
		
	\end{theorem}

	Finally we can show the existence of large expanders:
	\begin{corollary}
    \label{cor:large-expanders}
		If $\cC$ is a 
            \begin{enumerate}
                \item classical $[n,k,d]$ code then, there exists $\{K_i\}_i, K_i \subset G$, with $G$ a connectivity graph for $\cC$, such that $\sum_i |K_i| \geq \frac{k}{2}$, $|K_i| \geq \frac{d}{3}$, and each $K_i$ is $\Omegalog\left({\frac{k}{n}}\right)$-expander

                \item quantum $[[n,k,d]]$ code then, there exists $\{K_i\}_i, K_i \subset G$, with $G$ a connectivity graph for $\cC$, such that $\sum_i |K_i| \geq \frac{k}{2}$, $|K_i| \geq \frac{d}{3}$, and each $K_i$ is $\Omegalog\left(\sqrt{\frac{k}{n}}\right)$-expander
            \end{enumerate}

	\end{corollary}
	\begin{proof}
		This is an adaptation of the proof of the main theorem of \cite{baspin2021quantifying}, combined with Theorem \ref{thm:no-expander-partitioning}.
	\end{proof}

\section{Converse to the BPT bounds for classical codes}
\label{sec:converse-bpt}


In this section we show how to saturate the BPT bound for classical codes, and this exposition serves as a warm up for the next section which will address spaces not as well behaves as $E^D$. The approach we consider is inspired from \cite{bacon2015sparse}: we start with an arbitrary $[n, \Theta(n), \Theta(n)]$ LDPC code -- which we call $\cC_{\text{in}}$-- and embed it into $\E^D$ while introduced a \emph{controlled} amount of overhead, which leads to a $\param{n}{n^{(D-1)/D}}{n}$ code.

We start by illustrating how our technique works in 2D, and then we sketch how it extends to arbitrary dimensions. We line up the $n$ bits in a 2D plane along the $\hat{y}$ direction, and consider a subset $\{c_i\}_i$ of checks whose support does not overlap. As pictured in Figure \ref{fig:bpt1}, those checks might connect bits that are very far apart. Denote $\omega_i$ the set of bits in the support of $c_i$, then it is possible to extend each of the bits in $\cup_i \omega_i$ into repetition codes along the $\hat{x}$ direction, and braid them such for each $i$,the bits in $\omega_i$ are close to each other -- in 2D, it is always possible to do this braiding with repetition codes of length $\Delta = \Theta(n)$ \cite{bacon2015sparse}. Once this is done it is possible to implement $c_i$ locally by acting on the repetition code extensions of the respective bits.

\begin{figure}[H]
    \centering
    \includegraphics[scale = 0.7, page=1]{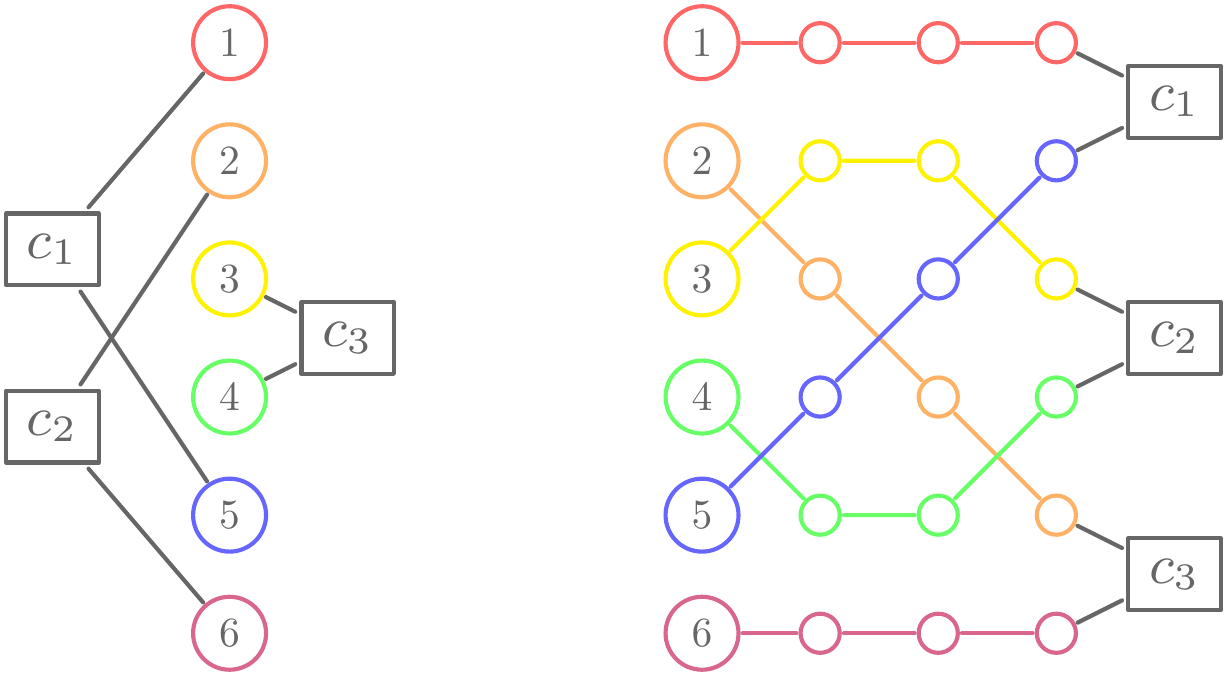}
    \caption{An example of how a code, on the left, with potential long range checks, can be mapped to a code with only 2D local checks on the right. In this case, $\omega_1$ contains bits $1$ and $5$, $\omega_2$ contains bits $2$ and $6$, and $\omega_3$ contains bits $3$ and $4$. The $\hat{y}$ direction goes from the bottom to the top of the page, while the $\hat{x}$ direction goes from left to right.}
    \label{fig:bpt1}
\end{figure}

\begin{figure}[H]
    \centering
    \includegraphics[scale = 0.7, page=2]{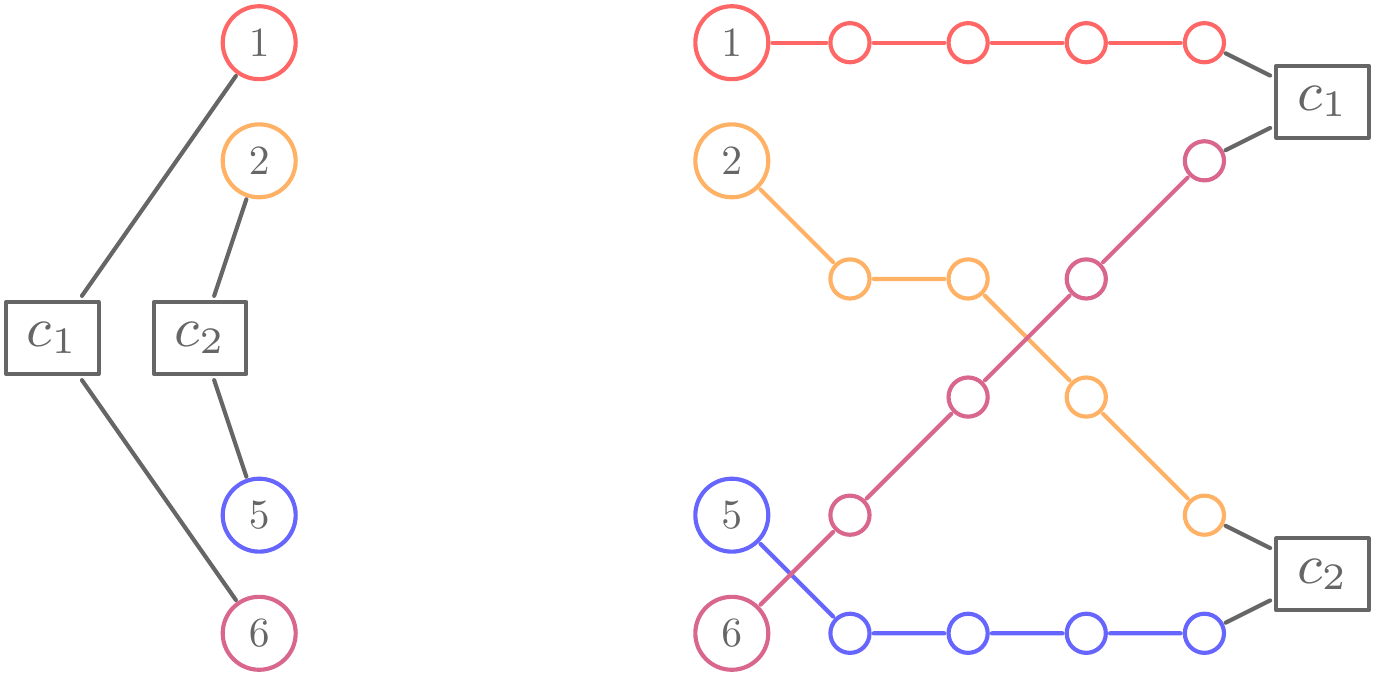}
    \caption{Example for a different code on the same bits.}
    \label{fig:bpt2}
\end{figure}

Now we deal with the situation where checks overlap. Since the code we consider is LDPC, it is always possible to partition the checks into a constant number of classes $\{\gamma_j\}_j$, $\gamma_j = \{c_{i,j}\}_i$ where for a given $j$, the support of the checks in $\gamma_j$ does not overlap. For each class $\gamma_j = \{c_{i,j}\}_{i}$, we repeat the same process outlined previously: denote $\omega_{i,j}$ the set of bits involved in the check $c_{i,j}$, then extend those bits into repetition codes and braid them so to bring the bits in $\omega_{i,j}$ close together; the checks $\{c_{i,j}\}_{i}$ can be enforced locally on the repetition codes. 
Due to the constant number of classes $\gamma_j$, this process is repeated a constant number of times, and yields a constant density of bits and checks. 

\begin{figure}[H]
    \centering
    \includegraphics[scale = 0.7, page=3]{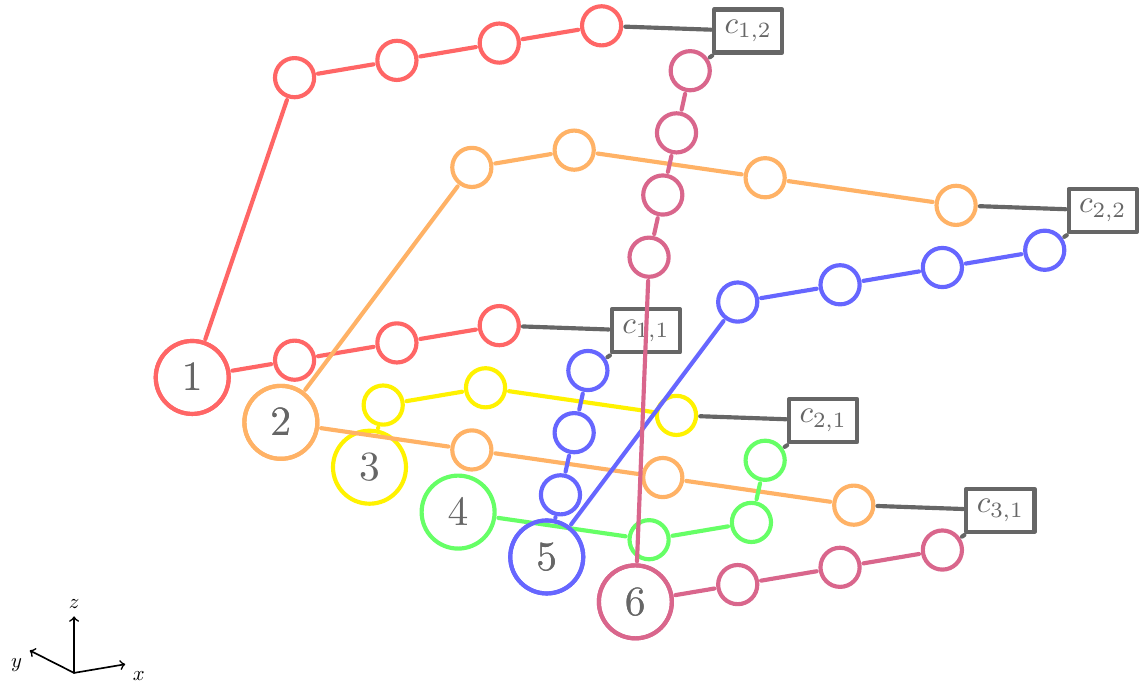}
    \caption{Result of the construction when given as input the same six bits, and the checks in Figure \ref{fig:bpt1} and \ref{fig:bpt2}. Since those checks overlap, we can take $\gamma_1 = \{c_1, c_2, c_3\}$, and $\gamma_2 =\{c_4,c_5\} $. On this figure, $c_1, c_2, c_3$ are relabeled as $c_{1,1}, c_{2,1}, c_{3,1}$, while $c_4$ and $c_5$ are relabeled as $c_{1,2}$ and $c_{2,2}$. Note how each $\gamma_i$ induces a new braiding layer in the $\hat{z}$ direction. When taking an LDPC code as an input, this number of layers stays constant, and as a result can be considered to live in 2D.} 
    \label{fig:bpt3}
\end{figure}

It is easy to verify that this code has $n' \sim n \cdot \Delta$ physical bits, $k' = k$ logical bits, and $d' \geq d \cdot \Delta$. This gives a $[n', \Theta(\sqrt{n'}),\Theta(n')]$ code local in 2D.

This construction can be extended to higher dimensions. We start by filling a $D-1$-dimensional cube with side length $n^{1/(D-1)}$ with the $n$ bits of $\mathcal{C}_{\text{in}}$. The braiding can then be performed in depth $\Delta = \Theta(n^{1/(D-1)})$ \cite{bacon2015sparse}. This gives a code with parameters $[n^{D/(D-1)}, \Theta(n), n^{D/(D-1)}]$, or equivalently, $[n', \Theta((n')^{(D-1)/D}), \Theta(n')]$.

\section{Partial converse to the expansion-based bound for classical codes}
\label{sec:converse-exp}
In this section, we show how, for any space in which a constant degree $n$-vertex $\alpha$-expander can be embedded, that space can also embed a $[{\Theta(n)},{\Omegalog(\alpha^2 n)},{\Omegalog(\alpha^2 n)}]$ classical code. We start by defining formally what a local embedding for a code means for a non-Euclidean space.

\begin{definition}
    \label{def:locally-embeddable}
    An $n$-bit code $\cC$ is said to be locally embeddable in a space $\cM = (X, d)$, if there exists a map $\eta: [n] \rightarrow X$, and universal constants $r,r'$ such that 

    \begin{enumerate}
        \item For any check $c$, the image of the bits involved in $c$ through the map $\eta$ are contained within a ball of radius $r$.
        \item Any ball of radius $r$ in $\cM$ contains the image of at most $r'$ bits.
    \end{enumerate}
\end{definition}

We might also want to define the notion of embedding for a graph. It is almost identical to the previous definition.

\begin{definition}
    \label{def:locally-embeddable-graphs}
    A graph $G = (V,E)$ is said to be locally embeddable in a space $\cM = (X, d)$, if there exists a map $\eta: V \rightarrow X$, and universal constants $r,r'$ such that 

    \begin{enumerate}
        \item For any edge $e$, the image of the vertices involved in $e$ through the map $\eta$ are contained within a ball of radius $r$.
        \item Any ball of radius $r$ in $\cM$ contains the image of at most $r'$ vertices.
    \end{enumerate}
\end{definition}

We can now state the main result of this section more formally.
\begin{proposition}

    Let a constant degree $\alpha$-expander graph $G$ be locally embeddable in a space $\cM$. Then there is a $[{\Theta(n)},{\Omegalog(\alpha^2 n)},{\Omegalog(\alpha^2 n)}]$ classical code that can be locally embedded in $\cM$.
		
\end{proposition}

\begin{definition}[Definition 8 of \cite{kleinberg1996short}]
    An immersion of a graph $H = (V_H, E_H)$ in $G = (V_G, E_G)$ is a collection of connected subgraphs of $G$, $\{X_v, v \in V_H\}$, and a collection of paths in $G$ $\{P_e, e \in E_H\}$ such that 

    \begin{enumerate}
    \item The subgraphs $\{X_v\}_v$ are mutually vertex disjoint
    \item If $e \in E_H$ has ends $u,w \in V_H$, then $P_e$ has its endpoints in $X_u$ and $X_w$, then it is vertex disjoint from all other $P_e'$, and all other $X_v$, besides $X_u$ and $X_w$.
    \end{enumerate}
\end{definition}

\begin{lemma}[\cite{chuzhoy2019large,krivelevich2019expanders} ]
    \label{lem:large-minors}
    There exists a universal constant $c$ such that every $n$-vertex $\epsilon$-expander graph of degree at most $\delta$ contains an immersion for any graph $H = (V_H, E_H)$ satisfying $|V_H|, |E_H| \leq c \frac{n}{\log n}\frac{\alpha^2}{\delta^2}$
\end{lemma}

We can use Lemma \ref{lem:large-minors} to embed $\param{n}{n}{n}$  codes into spaces at the cost of an overhead that only depends on some combinatorial property of the target space. This is in stark constrast with the embedding theorem from the previous section that strongly relies on the geometric structure of $\mathbb{E}^D$. 
For the sake of definiteness, the graph $H$ that will be immersed by Lemma \ref{lem:large-minors} is the Tanner graph of the input code $\cC_{\text{in}}$.

\begin{definition}
    Given a code $\cC$ and a set of checks $\{c_i\}_i$ for $\cC$, the Tanner graph $G$ has a vertex for every bit and every check, and there is an edge between a bit and a check if that bit is in the support of that check.
\end{definition}

It is easy to see that for an LDPC code, the Tanner graph has $O(n)$ vertices and edges. Using Lemma \ref{lem:large-minors} we can therefore embed the Tanner graph of an LDPC code of size $\Omega(\frac{n}{\log n}\frac{\alpha^2}{\delta^2})$. This immersion can then be lifted into a code by the following steps:

\begin{enumerate}
    \item All vertices in the subgraphs $\{X_v\}_v$ and the paths $\{P_e\}_e$ are mapped to bits.
    \item All the bits within a single $X_v$ are joined by a repetition code. The resulting code is local with respect to $G$ since $X_v$ is connected, i.e. there exists a path visiting every bit at least once.
    \item For every bit $q$ involved in a check $c$, with corresponding subgraphs $X_q$ and $X_c$, we use the path $P_e$ associated with the edge $(q,c)$ to extend $X_q$ into a repetition code all the way to $X_c$. At the boundaries between $X_q$ and $P_e$ we use a check between a bit in $X_q$, and a bit in $P_e$, and at the junction of $P_e$ and $X_c$ we use a check acting on two bits of $X_c$, and one bit of $P_e$.
\end{enumerate}

By adapting the argument from the previous section, it is straightforward to verify that $k,d \in \Omega(\frac{n}{\log n}\frac{\alpha^2}{\delta^2})$. Given that $n' \in \Theta(n)$, this gives a code with parameters $[n', \Omega(\epsilon^2 \frac{n'}{\log n'})],\Omega(\epsilon^2 \frac{n'}{\log n'})]$. Further, since $G$ is embeddable in $\cM$, then so is the resulting code. 

Those parameters however fall short of matching the bound \ref{cor:large-expanders}, which would require parameters scaling as $[n', \Omega(\epsilon n'), \Omega(n')]$ to be saturated.

\section*{Acknowledgements}
The author would like to thank Anthony Leverrier, Chinmay Nirkhe, Anirudh Krishna for many insightful comments.

N.\,B. is supported by the Australian Research Council via the Centre of Excellence in Engineered Quantum Systems (EQUS) project number CE170100009, and by the Sydney Quantum Academy.

\newpage

\printbibliography


\end{document}